\documentclass[preprint]{elsarticle}

\usepackage{lineno,hyperref}
\usepackage{amssymb}
\modulolinenumbers[5]

 \newtheorem{theorem}{Theorem} [section]
 \newtheorem{lemma}[theorem]{Lemma} 
 \newtheorem{corollary}[theorem]{Corollary}
 \newtheorem{claim}{Claim} 
 \newdefinition{definition}{Definition} 
 \newproof{proof}{Proof}
\newproof{pot}{Proof of Theorem \ref{thm2}}

\usepackage{graphicx}
\usepackage{enumitem}

\usepackage{algorithmic}

\usepackage{algorithm}

\begin{document}

\begin{frontmatter}

\title{Red Domination in Perfect Elimination Bipartite Graphs}

\author{Nesrine Abbas}

\address{MacEwan University, 5-173, 10700 104 Avenue, Edmonton, AB, T5J 4S2, Canada}

\ead{abbasn3@macewan.ca}

\begin{abstract}
The $k$ red domination problem for a bipartite graph $G=(X,Y,E)$ is to find a subset $D \subseteq X$ of cardinality at most $k$ that dominates vertices of $Y$. The decision version of this problem is $\textsc{NP}$-complete for general bipartite graphs but solvable in polynomial time for chordal bipartite graphs. We strengthen that result by showing that it is $\textsc{NP}$-complete for perfect elimination bipartite graphs. 
We present a tight upper bound on the number of such sets in  bipartite graphs, 
and show that we can calculate that number in linear time for convex bipartite graphs. We present a linear space linear delay enumeration algorithm that needs only linear preprocessing time.
\end{abstract}

\end{frontmatter}

\section{Introduction}

A red dominating set in a bipartite graph $G=(X,Y,E)$ is equivalent to a set cover where each set is represented by a vertex in $X$, vertices of $Y$ represent elements of the universe, and an edge between $x \in X$ and $y \in Y$ is added to the graph if the corresponding set contains the corresponding element. The set cover problem is to identify the smallest collection of subsets whose union equals the universe. The decision version of the set cover problem is $\textsc{NP}$-complete~\cite{Karp72}. There are well-known greedy approximation algorithms  to compute an approximate set cover, see~\cite{Chvatal79} for example. The aim of this paper is to study the problem for special bipartite graphs. The term red dominating set is better suited, and is what has been used in the literature, for graph contexts. The name originates from the view of a graph as partitioned into $Red$ and $Blue$ sets. 
The problem has been studied for general graphs in~\cite{AbuMouLied11}  where the authors develop an exact exponential-time algorithm for connected red dominating set.
The problem has practical application in the assignment of trains to stations. In~\cite{Weihe98}, the author develops an algorithm based on real data to find a subset of the stations that covers all trains.
In this paper we prove that the decision problem is $\textsc{NP}$-complete for perfect elimination bipartite graphs. We develop algorithms for counting and enumerating minimum cardinality red dominating sets in convex bipartite graphs. We present a tight upper bound on the number of such sets in bipartite graphs.

Convex bipartite graphs arise in scheduling problems. In~\cite{DekSah84}, for example, a convex bipartite graph represents jobs in partition $X$ and time slots in partition $Y$.  A red dominating set on such a graph is a subset of jobs that covers the time range. 
The red dominating set problem has been studied for chordal bipartite graphs, which properly include convex bipartite graphs, in~\cite{Golo16}. In that paper the authors present an $\mathcal O(n \cdot \min (|E| \cdot \log{n}, n^2))$ delay enumeration algorithm. Other domination, counting, and enumeration problems have been studied for convex bipartite graphs and classes of bipartite graphs that include them. Enumerating spanning trees in bipartite graphs has been studied in~\cite{ZhouBu20}. The authors in~\cite{LinChen17} present a linear time algorithm to count independent sets in tree-convex graphs. Counting independent sets in other bipartite graph subclasses that include convex bipartite graphs has been studied in~\cite{Lin18}. In~\cite{DamaMullKrat90}, the authors show that various domination problems are solvable in polynomial time for convex bipartite graphs. 

The next section contains definitions that will be used throughout the paper. 
We will then prove that the red dominating set decision problem is $\textsc{NP}$-complete for perfect elimination bipartite graphs. 
In the section that follows we present a tight upper bound on the number of vertices in a red dominating set in a  bipartite graph, 
show how to calculate the number of such sets in linear time for convex bipartite graphs and develop a linear space linear delay enumeration algorithm. 
We then conclude the paper.

\section{Definitions}
A graph $G=(V,E)$ in this paper is finite, connected, undirected, with no loops and no parallel edges. We will use $n$ to denote $|V|$. The \textit {neighbourhood set of vertex $v$} is $N(v)=\{u: uv \in E\}$. The {\textit degree of $v$} is $deg(v)=|N(v)|$. 
The graph induced by a subset $D \subseteq V$ is denoted by $G[D]$. $G[D]$ (or $D$) is an \textit {independent set} if it induces a graph that has no edges. 
A set of vertices $D \subseteq V$ \textit {dominates} another set $U \subseteq V$ if $\forall u \in U$, there is $v \in D$ such that $u=v$ or $uv \in E$. When $U$ is a single vertex $u$, we say that $D$ dominates $u$. A set $D \subseteq V$ is \textit {dominating} if it dominates $V$. 
A \textit {path} of length $k$, or a $k$-path, in $G$ is a sequence of distinct vertices $(u_{1},\ldots,u_{k+1})$ such that $u_{j} u_{j+1}\in E$, $\forall j=1,\ldots,k$. 
A \textit {connected graph} is one that has a path between each pair in its vertex set. 
A \textit {cycle} of length $k$ in $G$ is a sequence of distinct vertices $(u_{1},\ldots,u_{k})$ such that $u_{j}u_{j+1} \in E$ $\forall j=1,\ldots,k-1$, $u_{1}u_{k} \in E$. A \textit {chord} in a cycle is an edge joining non-consecutive vertices.

A \textit {bipartite graph} $G=(X,Y, E)$ is one whose vertex set $X \cup Y$ can be partitioned into two independent sets $X$ and $Y$. We will denote $|X|$ by $n_X$ and $|Y|$ by $n_Y$. A \textit {complete bipartite graph $K_{n_X,n_Y}$} contains all possible edges between its two partitions. 
An edge $xy$ is \textit{bisimplicial} if $N(x) \cup N(y)$ induces a complete bipartite graph. Let $\sigma=(e_i,\ldots,e_p)$, where $e_i=x_{i}y_{i}$, be an ordering of pairwise disjoint edges of $G$, $S_j=\{x_1,x_2,\ldots,x_j\} \cup \{y_1,y_2,\ldots,y_j\}$, and $S_0=\emptyset$. $\sigma$ is said to be a \textit{perfect edge elimination scheme} if $e_{j+1}$ is bisimplicial in $G[(X \cup Y) \setminus S_j]$, for $j=0,\ldots,p-1$, and $G[(X \cup Y) \setminus S_p]$ has no edges. A bipartite graph that admits a perfect edge elimination scheme is called a
\textit{perfect elimination bipartite graph}. 
A bipartite graph $G=(X,Y,E)$ is said to be \textit {chordal bipartite} if each cycle of length greater than $4$ has a chord. 
A bipartite graph $G=(X,Y,E)$ is said to be \textit {$Y$-convex} if vertices of $Y$ can be ordered so that for each $x \in X$ neighbours of $x$ appear consecutively in $Y$. We will simply refer to that graph as \textit {convex bipartite}. Such an ordering is called a \textit {convex ordering} and can be calculated in linear time~\cite{BoothLue76}. Figure \ref{fig:convex_graph} shows a convex bipartite graph. 
Convex bipartite graphs are a proper subset of chordal bipartite graphs~\cite{BranLeSpin99} which in turn are properly contained in perfect elimination bipartite graphs~\cite{Golumbic80}.

Let the neighbourhood set of some vertex $x, x \in X$, in the convex bipartite graph $G=(X,Y,E)$ whose vertices are in convex ordering be $N(x)=\{y_{a}, y_{a+1}, \ldots,y_{a+b}\}, y_{a} < \ldots <y_{a+b}$. Then $left(x)=y_{a}$ is the leftmost neighbour of $x$, and $right(x)=y_{a+b}$ is the rightmost neighbour of $x$. The \textit {neighbourhood array $\mathcal N$} stores the two values $left(x)$ and $right(x)$ for each vertex $x \in X$, i.e., $\mathcal N[i]=(left(x_i),right(x_i))$. Given a convex ordering, values of $left(x)$ may be calculated in $\mathcal O(n+|E|)$ time using a simple procedure that starts at $y=y_1$, moves sequentially in $Y$, and assigns the smallest value to $left(x)$ for all $x \in N(y)$. Values of $right(x)$ may be calculated similarly. We will use $\mathcal N$ as the data structure to represent a convex bipartite  graph in our algorithms.

A \textit {lexicographic convex ordering (lex-convex)} is a convex ordering and for $x_i, x_j$, $i < j$ if $left(x_i )< left(x_j)$ or $left(x_i)=left(x_j)$ and $right(x_i) < right(x_j)$. Cases where $left(x_i)=left(x_j)$ and $right(x_i) = right(x_j)$, i.e., $N(x_i) = N(x_j)$, can be numbered arbitrarily without disturbing other vertices' order. Given the array $\mathcal N$, a simple LSD (Least Significant Digit) radix sort may be used to obtain a lex-convex ordering of the vertices of $X$ in $\mathcal O(n)$ time. The keys will be  $left(x)$$right(x)$ for each $x \in X$. 
The ordering of the vertices of the graph in Figure~\ref{fig:convex_graph} is lex-convex. We assume a lex-convex ordering for all convex bipartite graphs in this paper.

The reader is referred to~\cite{Berge85} and~\cite{Golumbic80} for any missing graph theory definitions or notations.

We formally define the \textit {$k$ RED DOMINATION} decision problem and its minimum cardinality version.

\begin{definition}
\textit {$k$ RED DOMINATION}: Given a bipartite graph $G=(X,Y,E)$ and integer $k$, $1 \leq k \leq |X|$, is there a subset $D \subseteq X$ of cardinality at most $k$ that dominates $Y$, i.e., $\forall y \in Y$, does there exist some $x \in D,$ s.t. $xy \in E$?
\end{definition}

\begin{definition}
\textit {MCRD (Minimum Cardinality Red Domination)}: Given a bipartite graph $G=(X,Y,E)$, find a minimum cardinality red dominating set in $G$.
\end{definition}

\begin{figure}
\begin{center}
\includegraphics[width=3.3in, height=2.1in]{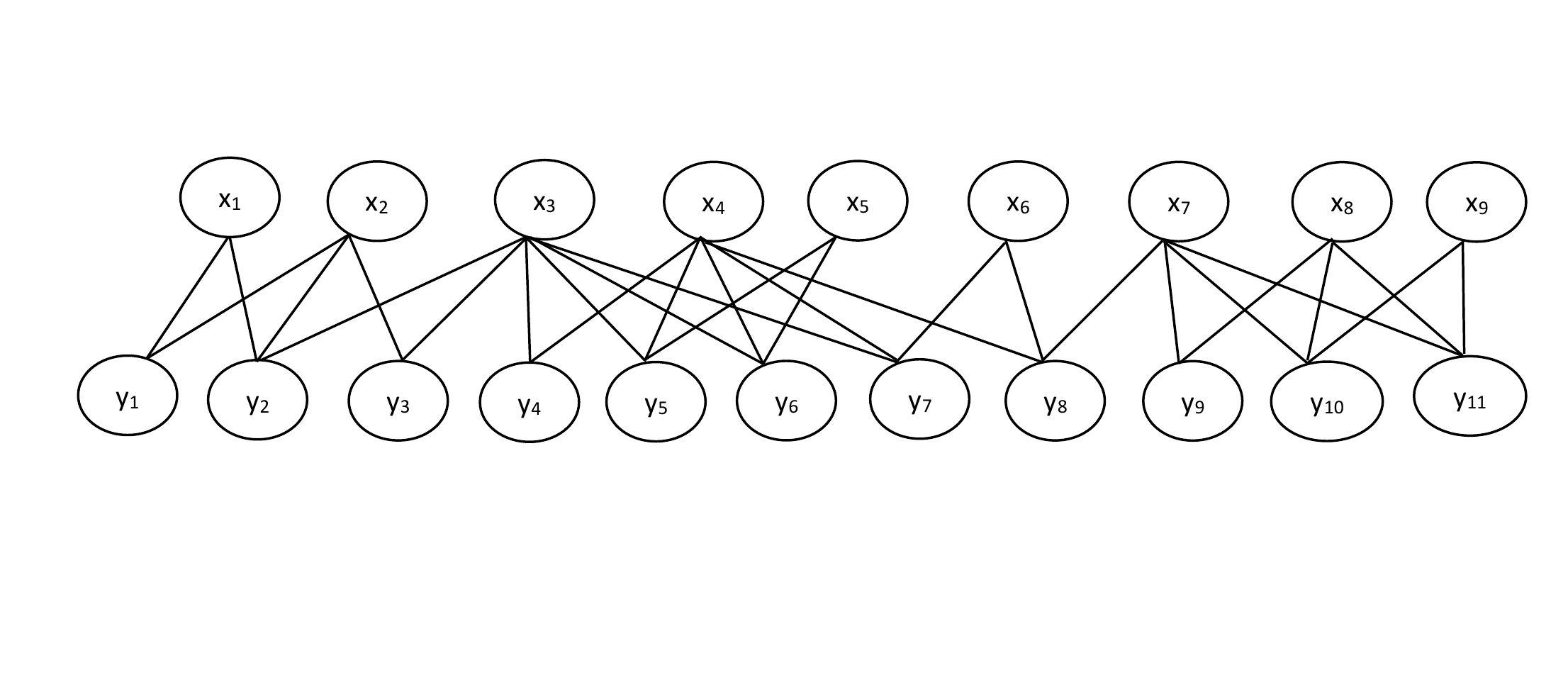}
\caption{A convex bipartite graph.}
\label{fig:convex_graph}
\end{center}
\end{figure}

\section {$\textsc{NP}$-completeness for Perfect Elimination Bipartite Graphs}
\label{sec:prel}

\begin{theorem}
$k$ RED DOMINATION is $\textsc{NP}$-complete for perfect elimination bipartite graphs.
\label{the:NP-c}
\end{theorem}
\begin{proof}
Clearly $k$ RED DOMINATION is in $\textsc{NP}$ because given a set $D$, we can verify in polynomial time if it is a solution for $k$ RED DOMINATION for the bipartite graph $G=(X,Y,E)$. 

We complete the proof by reduction from $k$ RED DOMINATION for bipartite graphs, which is $\textsc{NP}$-complete~\cite{Karp72}.
Given a connected bipartite graph $G = (X,Y,E)$ and a positive integer $k, 1 \leq k \leq n_X$, we construct a perfect elimination bipartite graph $G_P=(X_P,Y_P,E_P)$
such that $G$ has a red dominating set of cardinality $k$ if and only if $G_P$ has a red dominating set of cardinality $k$.
Let $Y_P=Y$, $X_P=X \cup T$ where $T=\{t_i: y_i \in Y\}$, and $E_P=E \cup \{t_{i}y_{i}, i=1,\ldots,n_Y\}$, i.e., we include all edges in the original graph and attach a vertex $t$ to each vertex $y$.
Clearly $G_P$ is bipartite and can be constructed in polynomial time. $G_P$ has the perfect edge elimination scheme $\sigma=(t_{1}y_{1},t_{2}y_{2},\ldots,t_{n_{Y}}y_{n_{Y}})$.

Let $D$ be a red dominating set in $G$. Since $Y_P=Y$, then $D$ induces a red dominating set in $G_P$.

Now let $D$ be a red dominating set in $G_P$. If $D$ contains no vertices from $T$, then $D$ induces a red dominating set in $G$. If $D \cap T \neq \emptyset$, then there is a vertex $x$ in $X_P \setminus T$ such that $N(x) \cap N(t) \neq \emptyset$ that may replace each such vertex $t \in D \cap T$ because we started with a connected graph $G$. Such a set $D$, with no vertices from $T$, will correspond to a red dominating set with the same cardinality in $G$. $\square$\end{proof}

\section {Counting and Enumerating All MCRD Sets for a Convex Bipartite Graph} 
In this section, we show a tight upper bound on the number of MCRD sets of cardinality $k$ in an arbitrary  bipartite graph and show that there are convex bipartite graphs that possess that number of MCRD sets. We present an algorithm that enumerates all MCRD sets in an arbitrary convex bipartite graph $G=(X,Y,E)$ in linear space and linear delay after a linear time preprocessing step. The algorithm operates in two stages. The first is a labelling preprocessing stage. The second is a branch and bound search on the labelled vertices of $X$ to output the MCRD sets. The labelling algorithm has been extended to calculate the number of MCRD sets. 

\begin{theorem}
Let $d$ be the maximum degree of vertices in $Y$ in the bipartite graph $G=(X,Y,E)$, $d=\max_{y \in Y} deg(y)$, $d \leq n_X$. The number of MCRD sets in an arbitrary  bipartite graph $G=(X,Y,E)$ is $\mathcal O(d^{k})$, where $k$ is the cardinality of an MCRD set.
\label{the:num-sets}
\end{theorem}
\begin{proof}
Let $d$ be the maximum degree of vertices in $Y$  in the  bipartite graph $G=(X,Y,E)$, $d=\max_{y \in Y} deg(y)$, $d \leq n_X$. Choose an arbitrary vertex $y_1 \in Y$. There are $deg(y_1) \leq d$ choices for the first vertex in an MCRD set  for $G$. 
Given vertices $u_1, u_2, \ldots u_a, a \geq 1$, in an MCRD set $D$, we only need to consider a vertex from $N(y)$ for an arbitrary vertex $y \in Y$ that is not dominated by  $u_1, u_2, \ldots, u_a$ for the next vertex in $D$. At each step in the process to build $D$, at most $d$ vertices need to be considered.
Using this process, some of the sets obtained may not be MCRD sets of cardinality $k$, however all  MCRD sets of cardinality $k$ will be generated. Hence the number of MCRD sets is $\mathcal O(d^k)$. $\square$\end{proof}

Figure~\ref{fig:enum} shows a convex bipartite graph that possesses that upper bound of MCRD sets. Not all vertices and edges are shown in the figure. The graph consists of $k, k>1$, complete bipartite graphs $K_{d-1,d-1}$, with an edge joining pairs of consecutive complete bipartite graphs {\textemdash} the thick edges shown in the figure. The maximum degree of vertices in $Y$ is $d$ and it is that of vertices incident with edges joining the complete bipartite graphs, shown filled with black in the figure. Any MCRD set will contain exactly one vertex from each complete bipartite graph and will be of cardinality $k$. Thus the number of MCRD sets is $(d-1)^k=\mathcal O(d^k)$. It is worth noting that we can construct a graph with $d^k$ MCRD sets by removing the thick edges joining the complete bipartite graphs but this will render the graph disconnected.

\begin{figure}
\begin{center}
\includegraphics[width=4in, height=2.3in]{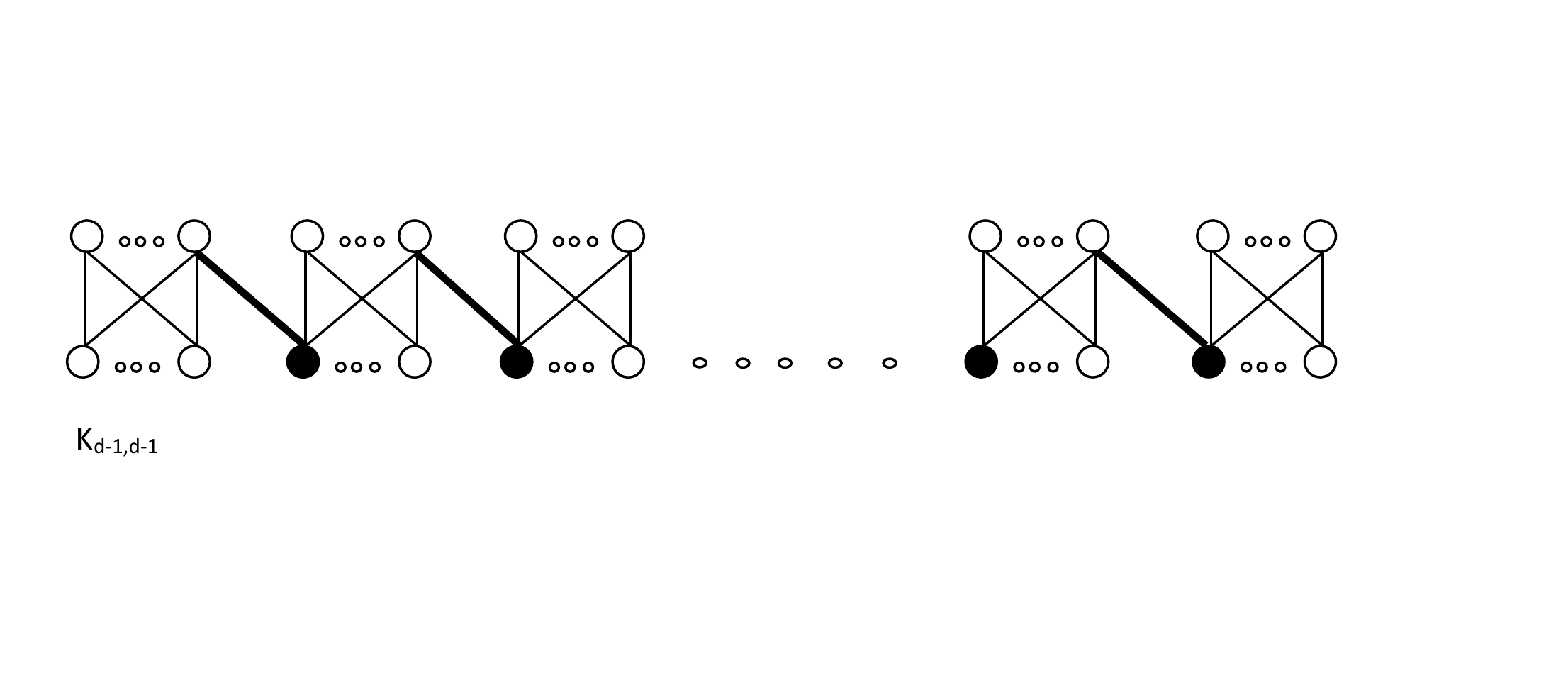}
\caption{A convex bipartite graph with $\mathcal O(d^k)$ MCRD sets.}
\label{fig:enum}
\end{center}
\end{figure}

In the next lemma we prove some algorithmic properties for MCRD sets in convex bipartite graphs.
\begin{lemma}
\label{lem:dom-neighb}
Let $D=\{u_{1}, u_{2}, \ldots,u_{k}\}$, $u_{1} < u_{2}< \ldots u_{k}$, be an MCRD set of cardinality $k$ for the convex bipartite graph $G=(X,Y,E)$. The following is true.
\begin{itemize}
\item For distinct vertices $u,v \in D$,  it cannot be the case that $left(u) \leq left(v) \leq right(v) \leq right(u)$, i.e., $N(v) \not\subseteq N(u), \forall u,v \in D$, $u \neq v$. 
\item For all consecutive pairs $u,v , u< v$ in $D$, $left(u)<left(v) \leq right(u)+1 \leq right(v)$, i.e., $N(u) \cup N(v)$ consists of consecutive vertices in $Y$.
\item For each vertex $u \in D$, $D$ cannot contain more than one vertex $v>u$ such that $left(u)<left(v) \leq right(u)+1 \leq right(v)$.
\end{itemize}
\end{lemma}
\begin{proof}

Let $D$ be an MCRD set of cardinality $k$ for the convex bipartite graph $G=(X,Y,E)$.
$N(v) \subseteq N(u)$ for some pair $u,v \in D$, i.e., $left(u) \leq left(v) \leq right(v) \leq right(u)$, implies $D\setminus\{v\}$ is a red dominating set of cardinality $k-1$, which contradicts the minimality of $D$.

Note that this implies that for all vertices $u,v \in D$, $u \neq v$, $left(u) \neq left(v)$ and $right(u) \neq right(v)$.

The union of the neighbourhood sets, $N(u) \cup N(v)$, of two consecutive vertices $u,v, u<v$ in a red dominating set in a convex bipartite graph whose vertices are in lex-convex ordering must consist of consecutive vertices in $Y$,i.e., $right(u)+1 \in N(v)$,  for otherwise some vertices in $Y$ would not be dominated. Therefore, $left(u) < left(v) \leq right(u)+1 \leq right(v)$ by the above paragraph.

Suppose there are vertices $u<v<w \in D$ such that $left(u) < left(v) \leq right(u)+1 \leq right(v)$ and $left(u) < left(w) \leq right(u)+1 \leq right(w)$. Then $N(u) \cup N(v) \subseteq N(u) \cup N(w)$ if $right(v) < right(w)$, and $N(u) \cup N(w) \subseteq N(u) \cup N(v)$ otherwise. In either case one of $v$ or $w$ may be removed from $D$ to produce a $k-1$ red dominating set, contradicting the minimality of $D$. $\square$\end{proof}

The enumeration algorithm operates in two stages. The first is a preprocessing labelling procedure that assigns labels to vertices of $X$. The second stage is a fast branch and bound search on vertices of $X$ that starts at a vertex adjacent to $y_{n_{Y}}$ with a particular label and proceeds backwards to a vertex adjacent to $y_1$. The preprocessing stage is presented in Algorithm~\ref{alg:bft}. Vertices adjacent to $y_1$ are assigned label $0$ and other vertices start with an initial label of $\infty$. Subsequently, a vertex $v$ is assigned label $a$ if the union of its neighbourhood set with that of a smaller vertex $u$ with label $a-1$ consists of consecutive vertices in $Y$, $N(v) \not\subseteq N(u)$, and $N(u) \not\subseteq N(v)$. The label $a$ will be the smallest such label for $v$. We use a FIFO queue to process the vertices. The algorithm can also be used to calculate the number of MCRD sets in linear time as we will prove. The array $Count$ is used for that purpose.

\begin{algorithm}
\caption{Assigning labels and counts to vertices of $X$ in the convex bipartite graph $G=(X,Y,E)$}
\hspace*{\algorithmicindent}\textbf{Input: }{The neighbourhood array $\mathcal N$ of the convex bipartite graph $G=(X,Y,E)$} 
\\
\hspace*{\algorithmicindent}\textbf{Output: }{Arrays $Label, Count, |Label|=|Count|=n_X$. $Label[x]$ stores the label of $x$. $Count[x]$, for $x \in N(y_{n_{Y}})$, stores the number of MCRD sets containing $x$.}
\begin{algorithmic}[1]

\STATE {$Label[x] \leftarrow \infty$, $\forall x \in X$;}
\STATE {$Count[x] \leftarrow 0$, $\forall x \in X$;}
\STATE {$enqueue(Queue,N(y_1))$;}\COMMENT{enqueue neighbours of $y_1$ into queue in order}
\STATE {$Label[x] \leftarrow 0$, $\forall x \in N(y_1)$;}
\STATE {$Count[x] \leftarrow 1$, $\forall x \in N(y_1)$;}
\WHILE {$Queue$ is not empty}\label{lin:bft-Q}
\STATE {$x \leftarrow dequeue(Queue)$;}
\newline
\COMMENT {the next loop labels vertices $x_j>x$ whose neighbourhood forms a union of consecutive vertices with $N(x)$, $N(x_j) \not\subseteq N(x)$,  $N(x) \not\subseteq N(x_j)$}
\FORALL {$x_j \in \{x':x'> x, x' \leq x_{n_{X}},left(x) < left(x') \leq right(x)+1 \leq right(x'), Label[x'] \geq Label[x] +1\}$}\label{lin:bft-xj}
\IF {$Label[x_j] > Label[x] +1$}
\STATE $enqueue(Queue,x_j);$ \COMMENT{enqueue vertex only if it is not in queue}
\ENDIF
\STATE {$Label[x_j] \leftarrow Label[x]+1$;}
\STATE {$Count[x_j] \leftarrow Count[x_j]+Count[x]$;}
\ENDFOR
\ENDWHILE
\end{algorithmic}
\label{alg:bft}
\end{algorithm}

Running Algorithm~\ref{alg:bft} on the graph in Figure~\ref{fig:convex_graph} results in the array $Label$ with values $[0,0,1,1,\infty,2,2,2,\infty]$ and the array $Count$ with values $[1,1,2,1,0,2,3,1,0]$.

\begin{theorem}
Algorithm~\ref{alg:bft} runs in $\mathcal O(n_{X}+|E|)$ time.
\end{theorem}
\begin{proof}
The set built in line~\ref{lin:bft-xj} will be of size at most $deg(right(x)+1)$. Adding up those sets over all iterations and since each vertex in $X$ is visited at most once, we can see that this step is executed at most $|E|$ times. Since all other operations in the  while loop take constant time and operations outside the loop take at most $\mathcal O(n_{X})$ time, therefore, Algorithm~\ref{alg:bft} runs in $\mathcal O(n_{X}+|E|)$ time. $\square$\end{proof}

Let $k$ be the cardinality of an MCRD for a given convex bipartite graph $G=(X,Y,E)$. 
We prove the following properties for the output of Algorithm~\ref{alg:bft}.
\begin{lemma}
The following is true for the output of Algorithm~\ref{alg:bft}.
\label{lem:labelling}
\begin{enumerate}[label=(\alph*)]
\item Each vertex adjacent to $y_1$ retains a label of $0$, and each other vertex has a label larger than $0$. \label{sta:a}
\item If $u<v$ and $Label[u] \neq \infty$ then $Label[u] \leq Label[v]$. \label{sta:b}
\item Once a vertex is assigned a label $< \infty$, its label is not modified. \label{sta:c}
\item If $left(u)=left(v)$ and both $u, v$ have labels $\neq \infty$, then $Label[u]=Label[v]$.\label{sta:e}
\item A vertex $v$ is assigned label $a > 0$ if and only if there is a vertex $u$ whose label is $a-1$, $u<v$, $N(v) \cup N(u)$ consists of consecutive vertices,  $N(v) \not \subseteq N(u)$, and $N(u) \not \subseteq N(v)$. \label{sta:d}

\end{enumerate}
\end{lemma}
\begin{proof}
The first statement is easy to see. It is also easy to note that at the start of each iteration of the while loop the queue will contain only vertices with labels smaller than $ \infty$.

We prove statement~\ref{sta:b} by induction on the index of vertices in $X$. It is true for $x \in N(y_{1})$. Assume it is true for all $x < x_p$, $1 < p \leq n_{X}$. The first time the label of $x_p, x_p \not \in N(y_{1})$, is modified is when it is first added to the set built in line~\ref{lin:bft-xj}. Let $x_p$ be assigned label $a$ for the first time during an iteration when $x=x_i < x_p$, $Label[x_i]=a-1$. For $Label[x_p]$ to change, it must be added to a set built in line~\ref{lin:bft-xj} at an iteration when $x=x'$, $x_i < x' < x_p$.  Then  $Label[x_p] =a \geq Label[x']+1$ and $Label[x_i] \leq Label[x']$ by the induction hypothesis, imply $Label[x']=a-1$. Hence $Label[x_p]$ will not change, and $Label[x_p] \geq Label[x], \forall x<x_{p}$. Therefore, by mathematical induction, if $u<v$ and $Label[u] \neq \infty$ then $Label[u] \leq Label[v]$. 

The construct above also proves statement~\ref{sta:c}.

Let $u < v$ and $left(u)=left(v)$, then $N(u) \subseteq N(v)$. Suppose they don't have the same labels, then $N(u) \subset N(v)$, i.e., $right(u) < right(v)$, and $Label[u]=a < Label[v]$, by statement~\ref{sta:b}. This means that $u$ is added to the set built in line~\ref{lin:bft-xj} when $x=x_i$ and $Label[x_i]=a-1$. $left(x_i) < left(u)=left(v) \leq right(x)+1 \leq right(u) < right(v)$ and $Label[x_i]+1 \leq Label[u] < Label[v]$ imply that $v$ too is added to that set during the same iteration and will receive label $a$. By statement~\ref{sta:c}, the label of $v$ will not change. This contradicts our supposition. Therefore, $left(u)=left(v)$ and both $u, v$ have labels $\neq \infty$ implies $Label[u]=Label[v]$. This proves statement~\ref{sta:e}.

A vertex $v$ is assigned label $a$ when it is added to the set in line~\ref{lin:bft-xj}, when $x=u<v$ has label $a-1$ and $left(u) < left(v) \leq right(u)+1 \leq right(v)$. $N(v) \cup N(u)$ consists of consecutive vertices and $N(v) \not\subseteq N(u)$ because 
\newline 
$right(u)+1 \in N(v)$. $N(u) \not\subseteq N(v)$ 
because $left(u) < left(v)$. Since the label of $v$ will not change by statement~\ref{sta:c}, this proves the if-then part ($\Longrightarrow$ direction) in statement~\ref{sta:d}.

If there is a vertex $u$ with label $a-1$, and a vertex $v >u$, whose label has not been modified from $\infty$, such that $N(v) \cup N(u)$ consists of consecutive vertices,  $N(v) \not \subseteq N(u)$,  and $N(u) \not \subseteq N(v)$, then $left(u) < left(v) \leq right(u)+1 \leq right(v)$ and $v$ will be added to the set in line~\ref{lin:bft-xj} when $x=u$ and will receive a label of $a$.
This concludes the proof for statement~\ref{sta:d} and the lemma. $\square$\end{proof}

\begin{lemma}
Let $D=\{u_1,u_2,\ldots,u_k\}, u_1 < u_2, \ldots,< u_k$, be a subset of $X$ for the convex bipartite graph $G=(X,Y,E)$, and $Label$ be the output of Algorithm~\ref{alg:bft}. $D$ is an MCRD set if and only if all of the following is true.
\begin{itemize}
\item $u_k$ is adjacent to $y_{n_Y}$
\item $Label[u_i]=i-1$ for all $u_i \in D$
\item $k-1=\min_{x \in N(y_{n_Y})} Label[x]$, i.e., it is the smallest label of any vertex adjacent to $y_{n_Y}$
\item for any consecutive vertices $u_i,u_{i+1} \in D$, $i=1,\ldots,k-1$, $left(u_{i}) \leq left(u_{i+1})-1 \leq right(u_{i}) < right(u_{i+1})$
\end{itemize}
\label{lem:D-iff}
\end{lemma}
\begin{proof}

$\Longrightarrow$: By the definition of red domination and the lex convex ordering, $u_1$ is adjacent to $y_1$ and $u_k$ is adjacent to $y_{n_Y}$. By statement~\ref{sta:a}  and statement~\ref{sta:d} in Lemma~\ref{lem:labelling}, $Label[u_1]=0$ and $Label[u_i]=i-1$. By the minimality of $D$, $k-1=\min_{x \in N(y_{n_Y})} Label[x]$. By Lemma~\ref{lem:dom-neighb}, for consecutive vertices $u_i$ and $u_{i+1}$, $left(u_i) < left(u_{i+1}) \leq right(u_{i})+1 \leq right(u_{i+1})$. Therefore, $left(u_i) \leq left(u_{i+1})-1 \leq right(u_{i})< right(u_{i+1})$.

$\Longleftarrow$: If a set $D=\{u_1,\ldots,u_k\} \subseteq X$ satisfies the conditions, then it is easy to see that it is a red dominating set. If it is not an MCRD set then $k>p$, where $p$ is the cardinality of an MCRD set of $G$. Then there is a vertex $x \in N(y_{n_{Y}})$ with a label of $p-1$ by the first part of this proof. This contradicts the definition of $k$ in the hypothesis. $\square$\end{proof}

\begin{theorem}
The number of MCRD sets for a graph $G=(X,Y,E)$ is equal to $\sum_{x \in N(y_{n_Y})} Count[x]$, where  $Count$ is the output of Algorithm~\ref{alg:bft}.
\label{the:counting}
\end{theorem}
\begin{proof}

By Lemma~\ref{lem:D-iff}, it suffices to prove that the number of MCRD sets containing $u$ for each $u \in \{u: u \in N(y_{n_{Y}}), Label[u]=k-1\}$ is equal to  $Count[u]$.
We will prove the more general statement that the number of MCRD sets containing  $u$ for each $u \in N(y_{n_{Y}})$ is equal to  $Count[u]$. We will first prove the following claim.
\begin{claim}
If $D=\{u_1,.u_2,\ldots,u_k\}, u_1 < u_2, \ldots,u_k$, is an MCRD set in $G$, then each $D_i \subseteq D$, where $D_i=\{u_1,\ldots,u_i\}$, is an MCRD set in the connected subgraph $H_i$ induced by $\{x: x \in X, x \leq u_i\} \cup \{y: y \in N(x), x \leq u_i\}$.
\end{claim}
\begin{proof}
If not then there is a set $D' \subseteq X, D'=\{v_1,\ldots,v_j\}, v_1 < \ldots,<v_j \leq u_i$, such that $|D'| =j < i$ and $D'$ is an MCRD set in $H_i$.
Then $v_j $ is adjacent to $right(u_i)$  and, therefore, $N(u_i) \subseteq N(v_j)$, and $\{v_1,\ldots,v_j , u_{i+1},\ldots,u_k\}$ is a red dominating set for $G$ of a smaller cardinality than $D$ which contradicts the minimality of $D$. $\square$\end{proof}

We use induction on the index of a vertex $x_i \in X$ to prove that for all $x_i \in X$, $Count[x_i]$ equals the number of MCRD sets in $H_i$, where $H_i$ is defined as in the above claim.  
The hypothesis is true for all $x \in N(y_1)$, because $Count[x]=1$ for all $x \in N(y_1)$.
Assume the hypothesis is true for all $x < x_p$, $1 < p \leq n_{X}$. 
Let $D$ be any MCRD set that contains $x_p$, $D=\{u_1,\ldots,u_{p-1},u_p=x_p,\ldots,u_{k}\}$. By the above claim, each $D_i=\{u_1,\ldots,u_i\}, 1 \leq i \leq k$, is an MCRD set in the connected subgraph $H_i$.
The number of MCRD sets that contain $x_p$ in the connected subgraph $H_p$
is equal to the sum, over all $i$, of the number of MCRD sets $D_i=\{u_1, u_2, \ldots, u_i\}$ where $u_i < x_p$, $left (u_i) \leq left(x_p)-1 \leq right(u_i) < right(x_p)$, and $Label[x_p]=Label[u_i]+1$ by Lemma~\ref{lem:D-iff}. 
By the induction hypothesis, the number  of such sets, for each $i$, is equal to $Count[u_i]$. 
For all such vertices $u_i$ where $left (u_i) < left(x_p)  \leq right(u_i) +1 \leq right(x_p)$ and $Label[x_p]=Label[u_i]+1$, $x_p$ will be added to the set built in line~\ref{lin:bft-xj} in Algorithm~\ref{alg:bft}, when $x=u_i$. Hence $Count[x_p]$ is equal to the number of MCRD sets in the connected subgraph $H_p$.

Therefore, by mathematical induction, $Count[x]$ is equal to the number of MCRD sets that contain $x$, for all $x \in N(y_{n_{Y}})$. $\square$\end{proof}

\begin{corollary}
The number of MCRD sets for an arbitrary convex bipartite graph can be computed in $\mathcal O(n_{X}+|E|)$ time.
\end{corollary}


To enumerate all MCRD sets, we perform the second stage of the enumeration process. Algorithm~\ref{alg:enumeration} presented below, performs a branch and bound search on vertices of $X$ starting from a vertex with label $k-1$ that is adjacent to $y_{n_{Y}}$ going backwards up to a vertex with label $0$. In the process we only add to the current output $D$ vertices whose neighbourhood sets together with that of the previous vertex in $D$ make a set of consecutive vertices in $Y$ and are not contained in each other. Once we reach label $0$, we output the set. As we will show, unlike general recursive backtracking techniques, the branch and bound algorithm does not generate failed partial solutions.
Applying the algorithm to the graph in Figure~\ref{fig:convex_graph}, outputs the sets $\{x_8,x_4,x_2\}$, $\{x_7,x_4,x_2\}$, $\{x_7,x_3,x_2\}$, and $\{x_7,x_3,x_1\}$.

\begin{algorithm}
\caption{Enumerating all MCRD sets for the convex bipartite graph $G=(X,Y,E)$}
\hspace*{\algorithmicindent} \textbf{Input: }{The neighbourhood array $\mathcal N$ of the convex bipartite graph $G=(X,Y,E)$, array $Label$, $k-1$: as defined in Lemma~\ref{lem:D-iff}} 
\newline
\hspace*{\algorithmicindent} \textbf{Output: }{All MCRD sets of $G$} 
\begin{algorithmic}[1]
\FORALL {$x$ such that $x \in \{ x': x' \in N(y_{n_{Y}}), Label[x']=k-1\}$}\label{lin:y1}
\STATE {$D \leftarrow \{x\}$;}
\STATE {GENERATE$(D,x)$;}\label{lin:enum-D}
\ENDFOR
\STATE {GENERATE($D,x$)} \COMMENT {recursive function}\label{fun:generate}
\IF {$Label[x]=0$}
\STATE output $D$;
\RETURN
\ENDIF
\newline
\COMMENT {the next loop finds all candidate vertices to extend the MCRD set}
\FORALL {$x_j \in\{x': x' < x, left(x') \leq left(x) -1 \leq right(x')<right(x), Label[x']=Label[x]-1\}$}\label{lin:enum-xj}
\STATE $D \leftarrow D \cup \{x_j\}$;
\STATE {GENERATE$(D,x_j)$;}\label{lin:rec-D}
\ENDFOR
\RETURN
\end{algorithmic}
\label{alg:enumeration}
\end{algorithm}

\begin{theorem}
Algorithm~\ref{alg:enumeration} correctly enumerates all MCRD sets in linear space and linear delay.
\end{theorem}
\begin{proof}
We first show that each call to the function GENERATE in line~\ref{lin:enum-D} will result in at least one output. To see this we note that by statement~\ref{sta:d} in Lemma~\ref{lem:labelling}, the set  built in line~\ref{lin:enum-xj} is always non-empty, which implies that with each recursive call to GENERATE  the cardinality of $D$ increases.

The output of the algorithm is an MCRD set by Lemma~\ref{lem:D-iff}. The algorithm will output at least one MCRD set by the above paragraph

Suppose there is an MCRD set $T=\{t_1,\ldots,t_k\}$ of cardinality $k$ that is not output by the algorithm. Among all such sets, assume $T$ has the longest sequence of vertices that agrees with an output of Algorithm~\ref{alg:enumeration}, i.e., there is a red dominating set $D=\{d_1,\ldots,d_{k}\}$, that is output by the algorithm, $t_i=d_i, a \leq i \leq k$, $T$ and $D$ have the smallest such index $a$. Since the loop in line~\ref{lin:y1} goes through all neighbours of $y_{n_Y}$ that have label $k-1$, we may assume that $a \leq k-1$.

By Lemma~\ref{lem:dom-neighb}, $left(t_{a-1})<left(t_{a})=left(d_{a}) \leq right(t_{a-1})  +1  \leq right(t_{a})=right(d_{a})$. Therefore,  $left(t_{a-1}) \leq left(d_{a})-1 \leq  right(t_{a-1}) < right(d_{a})$ and $t_{a-1}$ will be added to the set built in line~\ref{lin:enum-xj} in the call to function GENERATE with $x=d_a=t_a$. Hence the algorithm will output an MCRD set that agrees more with $D$ contradicting our choice of $T$. Therefore, the algorithm outputs all MCRD sets.

The algorithm needs space to represent $G$ via the array $\mathcal N$, array $Label$, and each output $D$. Therefore, the space requirement is $\mathcal O(n_X)$.

The time between each output is the time to go from a vertex with label $k-1$ adjacent to $y_{n_Y}$ down to a vertex with label $0$ adjacent to $y_{1}$, choosing those vertices in line~\ref{lin:enum-xj}. Since with each recursive call to GENERATE in line~\ref{lin:rec-D} the cardinality of $D$ increases, therefore, the delay is $\mathcal O(k)$. $\square$\end{proof}

\section{Conclusion}
We have shown that the red dominating set problem is $\textsc{NP}$-complete for perfect elimination bipartite graphs. 
We have shown that  $\mathcal O(d^k)$ is a tight upper bound on the number of minimum red dominating sets in bipartite graphs, where $d$ is the maximum degree in $Y$ and $k$ is the minimum cardinality of a red dominating set. We have presented an algorithm to enumerate all minimum cardinality red dominating sets in convex bipartite graphs in linear space and linear delay. One advantage of that algorithm is that its preprocessing step takes only linear time. The preprocessing step has been used to calculate the number of minimum cardinality red dominating sets in a convex bipartite graph. 
An interesting extension to the results presented in this paper is to explore whether the two stage labelling branch and bound enumeration technique is general enough to be applied to other enumeration problems while maintaining linearity. Because convex bipartite graphs are not symmetric, it may be worthwhile studying the enumeration of blue domination for convex bipartite graphs.



\end{document}